\documentclass{llncs}
\usepackage[textsize=tiny]{todonotes}
\usepackage{xspace}
\usepackage{paralist}
\usepackage{graphicx}
\usepackage{subfig}
\usepackage{hyperref}
\usepackage{wrapfig}

\usepackage{amsmath,amssymb}
\usepackage[textsize=tiny]{todonotes}

\newcommand{\rephrase}[3]{\noindent\textbf{#1 #2}.~\emph{#3}}

\title{Strengthening Hardness Results to\\ 3-Connected Planar Graphs
}

\author{Giordano Da Lozzo\inst{1} \and Ignaz Rutter\inst{2}}

\institute{%
Department of Engineering, Roma Tre University, Italy
\and
Karlsruhe Institute of Technology, Germany}

\let\doendproof\endproof
\renewcommand{\endproof}{~\hfill$\qed$\doendproof}
\newcommand{\remove}[1]{\xspace}

 \newcommand{\eps}{\ensuremath{\varepsilon}}  
\newcommand{\mis}{{\sc MIS}\xspace}

  \DeclareMathOperator{\skel}{skel} 
  \DeclareMathOperator{\expd}{exp} 
  \DeclareMathOperator{\pert}{pert}

\makeatletter
\renewcommand{\todo}[2][]{\@bsphack\@todo[#1]{\textcolor{black}{#2}}\@esphack\ignorespaces}
\makeatother

\begin{document}

\maketitle

\begin{abstract}
In this paper we extend some classical NP-hardness results from the class of $2$-connected planar graphs to subclasses of $3$-connected planar graphs. The reduction are partly based on a new graph augmentation, which may be of independent interest.
\end{abstract}

\section{Introduction}
\label{sec:introduction}

When proving NP-hardness results for graph drawing problems on planar graphs with variable embedding it is often necessary to restrict the embedding choices of a construction.
In this case it is convenient to reduce from problems that are NP-complete for subclasses of $3$-connected planar graphs, which have an almost unique combinatorial embedding.
While there exist many hardness results for $2$-connected planar graphs, only few results are known for the $3$-connected case.

In this paper we strengthen some classical NP-hardness results to this setting. We show NP-hardness for {\sc Maximum Independent Set} ({\sc MIS}) for $3$-connected cubic planar graphs and planar triangulations, {$3$-Coloring} for $3$-connected bounded-degree planar graphs, and {\sc Steiner Tree} for $3$-connected cubic planar graphs.

The reductions for {\sc MIS} and {\sc 3-Coloring} are based on a new
graph augmentation technique (Lemma~\ref{lem:augmentation}) to
transform \mbox{$2$-connected} planar graphs with bounded degree into
$3$-connected planar graphs with bounded degree by ``subdividing''
each edge at most once, which may be of independent interest.  

We have recently used the hardness of \mis for 3-connected planar
graphs for showing hardness of an embedding problem that asks to
maximize the number of facial cycles that are contained in a given set
$\mathcal C$~\cite{dr-crfc-16}.  Similarly, our result on Steiner
trees can be used to simplify and extend
hardness results for embedding problems (e.g.,\cite[Theorem 8]{adfjk-tppeg-14},~\cite[Theorem 7]{adn-aspbe-15}, and~\mbox{\cite[Theorem 17]{addfpr-rccp-15})}.

\section{Preliminaries}
\label{sec:preliminaries}

%\subsubsection{Drawings and Embeddings.} 
We assume familiarity with basic concepts of graph drawing and planarity (see, e.g.,~\cite{dett-gd-99}).
% A \emph{planar drawing} $\Gamma$ of a graph maps vertices to points in the plane and edges to internally disjoint curves. Drawing $\Gamma$ partitions the plane into topologically connected regions, called  {\em faces}.
% The bounded faces are \emph{internal} and the unbounded face is the \emph{outer face}. 
% %A vertex (an edge) is {\em external} if it is incident to the outer face and {\em internal} otherwise.
% A planar drawing determines a circular ordering of the edges incident to each vertex. Two planar drawings of a connected planar graph are \emph{equivalent} if they determine the same orderings and have the same outer face. A \emph{combinatorial embedding} is an equivalence class of planar drawings.
%
For the definition of the SPQR-tree of a biconnected graph and the concepts of {\em skeleton} $\skel(\mu)$ and {\em pertinent graph} $\pert(\mu)$ of a node $\mu$ of an SPQR-tree, and that of {\em virtual edge} of a skeleton, and {\em expansion graph} of a virtual edge we refer the reader to~\cite{djkr-pesuf-14}.
For the definition of {\em canonical ordering} we refer the reader to~\cite{fpp-hdpgg-90}.
For convenience we also provide definitions in Appendix~\ref{apx:pre}.

Let $G=(V,E)$ be a plane graph with two designated edges $e',e'' \in
E$ incident to the same face.  A {\em graph augmentation} on the pair
$\langle e'=(u,v), e''=(w,z)\rangle$ turns $G$ into a new planar graph
$G'$ by replacing edges $e'$ and $e''$ with a connected planar graph
$G_A$ containing four vertices each of which is identified with one of
$\{u,v,w,z\}$ in such a way that $G_A$ is planar.  We say that a graph
augmentation is {\em $k$-good} if it does not increase the number of
$k$-cuts in the graph.  An {\em $H$-split} on the pair $\langle e',
e''\rangle$ is a graph augmentation that turns $G$ into a new plane
graph $G'$ by subdividing edges $e'$ and $e''$ with a dummy vertex
$v'$ and $v''$, respectively, and by adding edge $(v',v'')$.  Clearly,
an $H$-split is $2$-good.

\subsubsection*{NP-hard Problems.}
An independent set in a graph $G=(V,E)$ is a subset $V' \subseteq V$
of pairwise non-adjacent vertices.  The problem \mis asks for a
maximum size independent set.  A 3-coloring of a graph $G=(V,E)$ is an
assignment $c \colon V \to \{1,2,3\}$ such that for every edge $(u,v)
\in E$ it is $c(u) \ne c(v)$.  The problem {\sc 3-Coloring} asks
whether a given graph admits a 3-coloring.  let $(G,T)$ be a pair
where $G = (V,E)$ is a graph and $T \subseteq V$ is a set of
\emph{terminals}.  A Steiner tree is a subtree of $G$ that contains
all vertices in $T$.  The problem {\sc Steiner tree} asks for an
instance $(G,T)$ for a smallest Steiner tree of $G$, where the size is
measured in terms of the number of edges.

\section{Bounded-degree Augmentation}
\label{sec:augmentation}

In this section we give an algorithm (Lemma~\ref{lem:augmentation}) to augment a $2$-connected planar graph $G$ with minimum degree $\delta(G) \geq 3$ and maximum degree $\Delta(G)$ to a $3$-connected planar graph $G'$ with $\delta(G')=\delta(G)$ and $\Delta(G')=\Delta(G)$ by applying $H$-splits to disjoint pairs of edges of $G$.

A planar embedding of $\pert(\mu)$ is {\em regular} if the parent edge is incident to the outer face.
Let $\mathcal E_\mu$ be a planar embedding of $\pert(\mu)$ and let  $e$ be an edge of $\pert(\mu)$ that is incident to the outer face after removing the parent edge.
Embedding $\mathcal E_{\mu}$ is {\em $e$-externally $3$-connectible} if either
\begin{inparaenum}
\item $\mu$ is a Q-node or
\item the graph obtained from $\pert(\mu)$ by performing an $H$-split on $e$ and the parent edge is a subdivision of a $3$-connected planar graph whose only degree-$2$ vertices, if any, are the poles of $\mu$.
\end{inparaenum}
Also, we say that $\mu$ is {\em $e$-externally $3$-connectible} (or, simply, {\em externally $3$-connectible}) if $\pert(\mu)$ admits an $e$-externally $3$-connectible embedding, for some edge $e$ of $\pert(\mu)$.

Let $\mu$ be a node of the SPQR-tree $\mathcal T$ of $G$, let $e$ be an edge of $\pert(\mu)$, and let $\pert^*(\mu)$ be a graph obtained by applying $H$-splits on distinct pairs of edges in $\pert(\mu)$. 
We say that $e$ is a {\em free edge} if $e \in E(\pert(\mu)) \setminus E(\pert^*(\mu))$, that is, edge $e$ has not been used in any $H$-split.

Consider a non-Q-node $\mu$. Let $\pert^*(\mu)$ be a graph obtained from $\pert(\mu)$ via a set of edge-disjoint $H$-splits, let $L_\mu=[e_1,e_2]$ and $R_\mu=[e_3]$ be two lists of free edges in $E(\pert^*(\mu))$, and let $\mathcal E^*_\mu$ be a regular embedding of $\pert^*(\mu)$. 
We say that $\mathcal E^*_\mu$ is {\em extendible} if $L_\mu$ and $R_\mu$ are incident to different faces of $\mathcal E^*_\mu$ incident to the parent edge and $\mathcal E^*_\mu$ is $e$-externally $3$-connectible with $e \in L_\mu$.  

The proof is based on inductively constructing an extendible embedding of each node $\mu$ of the SPQR-tree assuming that extendible embeddings exist for the children of $\mu$. 
For S-nodes and P-nodes the construction is straightforward; see Fig.~\ref{fig:ext-SP-node}. Note that after the augmentation there are two free edges on one side of the embedding and one on the other side, which satisfies our invariant.

\begin{figure}[tb!]
\centering
  \subfloat{\includegraphics[page=1]{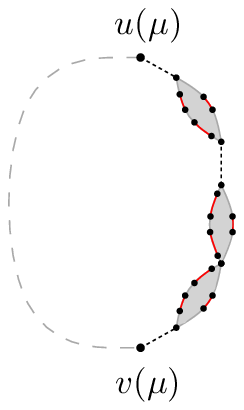}\label{fig:ext-S-node-a}}
  \subfloat{\includegraphics[page=2]{fig/ext-S-node}\label{fig:ext-S-node-b}}
  \subfloat{\includegraphics[page=1]{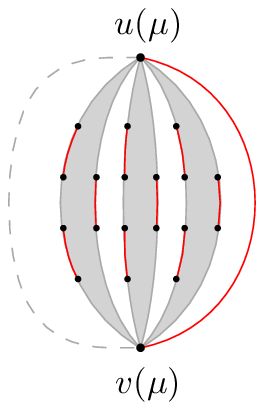}\label{fig:ext-P-node-a}}
  \subfloat{\includegraphics[page=2]{fig/ext-P-node}\label{fig:ext-P-node-b}}
\caption{
Augmentation of an $S$-node (a--b) and of a P-node (c--d) via $H$-splits.
% Illustration for the proof of Lemma~\ref{lem:augmentation} when $\mu$ is ((a),(b)) an S-node with three non-Q-node children $\nu_1$, $\nu_2$, and $\nu_3$, and ((c),(d)) a P-node with a Q-node child. 
% ((a), (c)) Auxiliary graphs obtained by replacing each virtual edge $\eps_i$ in $\skel(\mu)$ by an augmentation $\pert^*(\mu_i)$ of the corresponding child $\mu_i$. 
% ((b),(d)) Augmentation the auxiliary graph via $H$-splits to obtain $\pert^*(\mu)$.
Free edges are red, while the vertices and edges created by $H$-splits are green.
}
\label{fig:ext-SP-node} 
\end{figure}

For an R-node $\mu$ we construct an extendible embedding by processing the vertices of $\skel(\mu)$ according to a canonical ordering as illustrated in Fig.~\ref{fig:ext-R-node}. The construction allows us to satisfy our invariant by embedding the augmentations of the children of $\mu$ such that they contribute free edges to the outer face in each step of the canonical ordering. Special care has to be taken if the edge from which the canonical ordering starts does not correspond to a Q-node.

\begin{figure}[tb!]
  \centering
  \subfloat{ \includegraphics[height=0.16\textwidth,page=1]{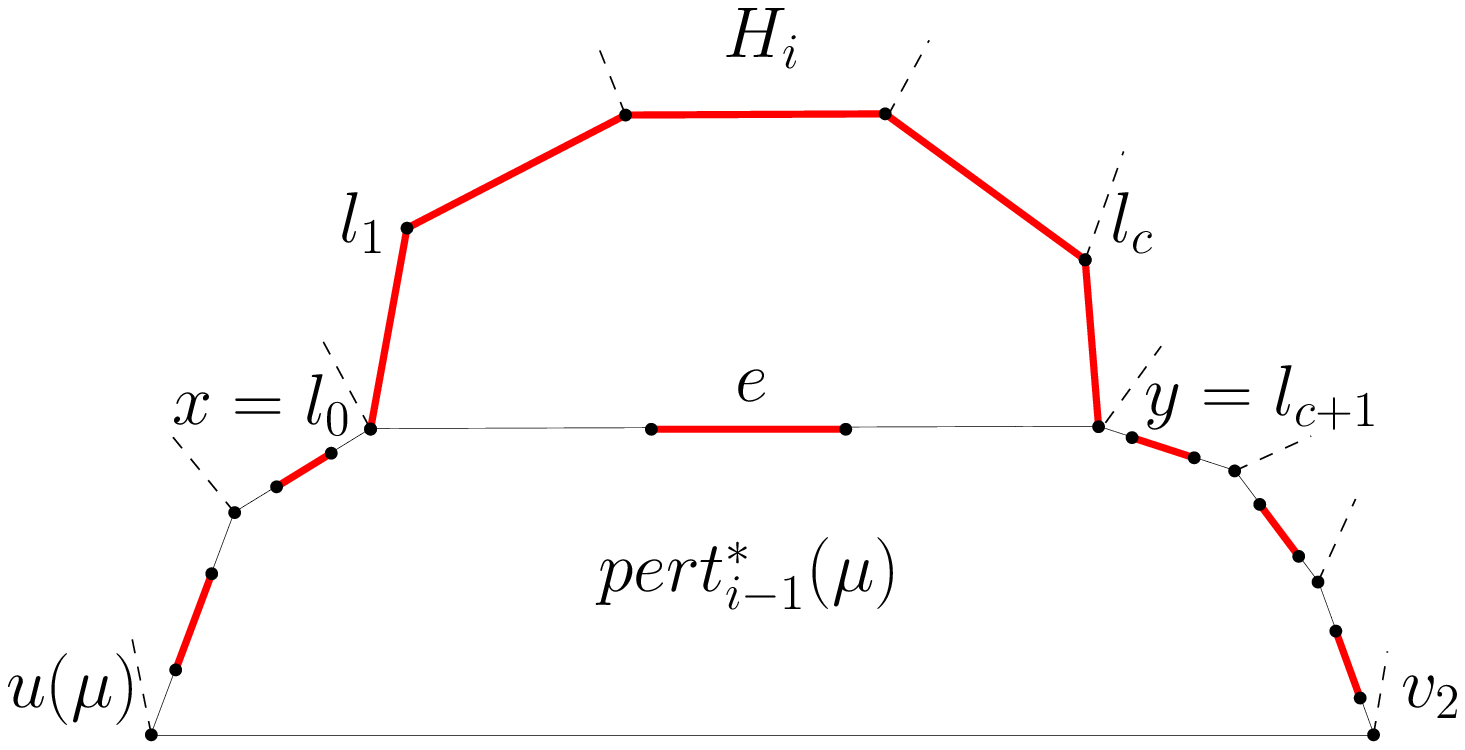}}
  \hfil
  \subfloat{ \includegraphics[height=0.16\textwidth,page=2]{fig/ext-R-node}}
  \hfil
  \subfloat{ \includegraphics[height=0.16\textwidth,page=3]{fig/ext-R-node}}
  \caption{
Augmentation of an R-node via $H$-splits for the next part of the canonical ordering. Free edges are red, vertices and edges created by $H$-splits are green.
  % Illustration for the proof of Lemma~\ref{lem:augmentation} when $\mu$ is an R-node. (a) Augmentation $\pert_{i-1}^*(\mu)$ of a part of the canonical ordering together with the next path of virtual edges of $\skel(\mu)$. (b) Auxiliary graph obtained from (a) by replacing each virtual edge $\eps_i$ in the path by an augmentation $\pert^*(\mu_i)$ of the corresponding child $\mu_i$. 
  %   (c) Augmentation the auxiliary graph via $H$-splits to obtain the augmentation $\pert_i^*(\mu)$ for the next part.
    }
  \label{fig:ext-R-node}
\end{figure}

Since the described augmentation only uses $H$-splits, it does not increase the minimum and maximum degree beyond $3$. We have the following main result.

\begin{lemma}
  \label{lem:augmentation}
  Let~$G=(V,E)$ be a $2$-connected planar graph with minimum degree $\delta(G)\geq 3$ and maximum degree $\Delta(G)$.  There exist disjoint pairs $(e'_1,e''_1)$, $\ldots$, $(e'_k,e''_k)$ of edges in $E$ such that performing the $H$-splits
  $\langle e'_1,e''_1\rangle$, $\ldots$, $\langle e'_k,e''_k\rangle$ yields a $3$-connected planar graph $G'$ with $\delta(G')=\delta(G)$ and $\Delta(G')=\Delta(G)$.
\end{lemma}

% In the next section we exploit Lemma~\ref{lem:augmentation} to prove NP-hardness of the {\sc Maximum Independent Set} problem for $3$-connected cubic planar graphs.
% In the Appendix\todo{Add a reference to the right appendix}, we will give further examples of application of
%Lemma~\ref{lem:augmentation} to extend NP-hardness results from the class of $2$-connected planar graphs with bounded degree to that of $3$-connected planar graphs with bounded degree.

\section{Hardness Results Based on $2$-Good Augmentations}
\label{sec:hardness-augmentation}

In this section we give examples on how to exploit Lemma~\ref{lem:augmentation} to extend NP-hardness results from the class of $2$-connected planar graphs with minimum degree $3$ and bounded maximum degree to that of $3$-connected planar graphs with bounded maximum degree.
The general line is as follows.
Given an NP-hard decision problem $\mathcal P$ which takes as input a planar graph $G$ and, possibly, a parameter $k$, one just needs to define a graph augmentation which
(i) is $2$-good
and (ii) replaces a pair of edges $\langle e',e''\rangle$ of $G$ with a gadget $A(\langle e',e''\rangle)$ of polynomial size to obtain a
graph $G'$ such that $(G,k)$ is a \texttt{yes} instance for problem $\mathcal P$ if and only if $(G',k')$ is a \texttt{yes} instance for problem $\mathcal P$,  where $k' = f(k)$ and $f$ is a computable polynomial function.

\subsubsection*{Maximum Independent Set.}

For a graph~$G$ we denote by~$\alpha(G)$ the size of a largest independent set in~$G$.

\begin{lemma}
  \label{lem:independent-increase}
  Let~$G$ be a $2$-connected cubic planar graph and let~$e$ and~$e'$ be two edges in~$E(G)$ incident to the same face of a planar embedding of $G$.  Let~$G'$ be the $2$-connected cubic planar graph obtained from $G$ by applying the graph augmentation illustrated in Fig.~\ref{fig:mis-cubic-gadget} to~$e$
  and~$e'$.  Then~$\alpha(G') = \alpha(G) + 5$.
\end{lemma}

\begin{proof}

  \begin{figure}[tb!]
    \centering
    \subfloat[]{
      \includegraphics[height=0.12\textheight]{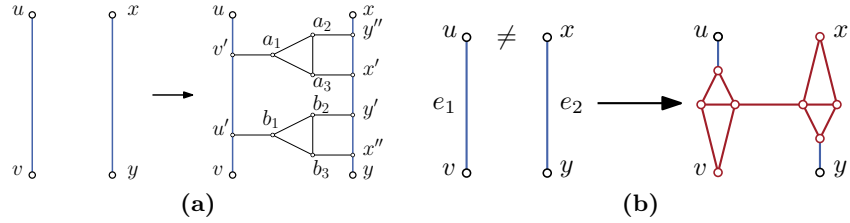}
      \label{fig:mis-cubic-gadget}
    }
    \hfil
    \subfloat[]{
      \includegraphics[page=2]{fig/A1}
      \label{fig:A1}
    }
    \caption{ (a) Gadget for the proofs of
      Lemma~\ref{lem:independent-increase}.  (b) Gadget $A_1(\langle
      e_1,e_2\rangle)$ used in the hardness proof of $3$-Coloring.
      Before the substitution the endpoints $u,v$ and $x,y$ of $e_1$
      and $e_2$, respectively, are exchanged so that $u \ne x$ holds.}
  \end{figure}

  Let~$I$ be an independent set in~$G$.  We construct an independent set $I'$ of~$G'$ as follows.  We start with~$I$.  If~$I$ does not contain~$u$, we add~$v'$ to it; otherwise it does not contain~$v$, and we can add~$u'$.  Similarly, if~$I$ does not contain~$x$, we add~$y',y''$ as well as~$a_3$
  and~$b_3$; otherwise it does not contain~$y$, and we add~$x',x''$ as well as~$a_2$ and~$b_2$.  Clearly~$I'$ is an independent set of size~$|I|+5$, showing that~$\alpha(G') \ge \alpha(G) + 5$.

  Conversely, assume that~$I'$ is an independent set of~$G'$.  First assume that~$a_1 \in I'$.  Note that~$I''$ cannot contain both~$x'$ and~$y''$.  It follows that we can replace~$a_1$ by either~$a_2$ or by~$a_3$ to obtain an independent set of the same size.  An analogous argument applies
  to~$b_1$.  We can hence assume without loss of generality that neither~$a_1$ nor~$b_1$ is contained in~$I'$.  Now assume that~$I'$ contains both~$u$ and~$v$.  It then follows that~$u'$ and~$v'$ are not in~$I'$, and hence~$(I' \setminus \{v\}) \cup \{u'\}$ is an independent set
  (since~$b_1 \not \in I'$) of the same size.  We can hence also assume that~$I'$ contains at most one vertex in~$\{u,v\}$.

  Next assume that~$\{x,y\} \subseteq I'$.  Then~$I'$ contains at most one vertex in~$\{x',y'$, $x'',y''\}$, at most one vertex in~$\{a_2,a_3\}$, and at most one vertex in~$\{b_2,b_3\}$.  Then~$I' \setminus (\{x,y\} \cup \{x',y',x'',y''\} \cup \{a_2,a_3\} \cup \{b_2,b_3\}) \cup \{x,x',x'',a_2,b_2\}$ is
  a larger independent set containing only one vertex in~$\{x,y\}$.  Hence, we can also assume without loss of generality that~$I'$ contains at most one vertex in~$\{x,y\}$.

  It follows from the above that, after suitably transforming~$I'$, the set~$I' \cap V(G)$ is an independent set in~$G$.  Now observe that~$I'$ can contain at most one vertex from each of the sets~$\{u',v'\}$, $\{a_1,a_2,a_3\}$, $\{b_1,b_2,b_3\}$, $\{x',y''\}$, and~$\{x'',y'\}$.  It follows
  that~$|I' \cap V(G)| \ge |I'| - 5$.  In particular, this implies~$\alpha(G) \ge \alpha(G') - 5$, or, equivalently,~$\alpha(G') \le \alpha(G) + 5$.
  % This concludes the proof.
\end{proof}

By applying Lemma~\ref{lem:independent-increase} to the distinct pairs of edges of a biconnected cubic planar graph (for which \mis is NP-complete~\cite{mohar-fcgpag-01}) determined by Lemma~\ref{lem:augmentation} we obtain the following.

\begin{theorem}
  \label{thm:mis-cubic}
  \mis is NP-complete for $3$-connected cubic planar graphs.
\end{theorem}

\begin{wrapfigure}[9]{r}{.4\textwidth}
\vspace{-0.55cm}
  \centering
  \includegraphics[height=0.15\textheight]{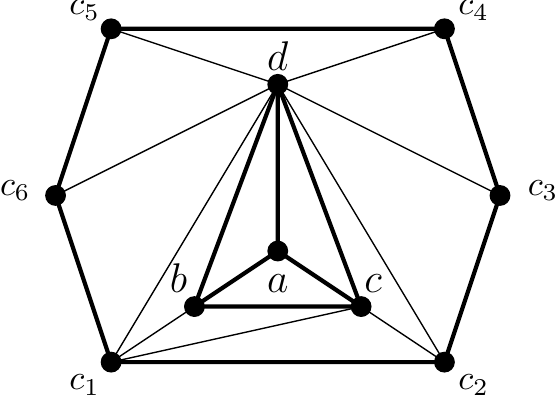}
  \caption{Gadget $\Phi_6$ for placement inside a face of size~$6$.}
  \label{fig:mis-triangulation-inner}
\end{wrapfigure}

\subsubsection*{$3$-Coloring.}

Let $\mathcal C(G)$ denote the number of cut vertices of a graph $G$.
Given a pair of edges $\langle e_1 = (u,v),e_2 = (x,y)\rangle$ of a
planar graph $G$, we define gadget $A_1(\langle e_1,e_2\rangle)$ as the
graph illustrated in Fig.~\ref{fig:A1}, where we assume that $u \neq
x$. Observe that gadget $A_1(\langle e_1,e_2\rangle)$ is $2$-good.  We
first prove an auxiliary lemma.

\begin{lemma}\label{lem:coloring-strenghten}
  The {\sc $3$-Coloring} problem is NP-complete for $2$-connected planar graphs with minimum degree $4$ and maximum degree $7$.
\end{lemma}

We can now exploit Lemma~\ref{lem:coloring-strenghten} and Lemma~\ref{thm:mis-cubic} and the fact that gadget $A_1$ is $2$-good to obtain the following theorem. % (the upper bound on the maximum degree derives from the fact that gadget $A_1$ increases the degree of the end vertices of the corresponding edge pair by at most $1$ and that
% Lemma~\ref{thm:mis-cubic} garantes that we never consider any edge for augmentation more than once).
  
\begin{theorem}
    The {\sc $3$-Coloring} problem is NP-complete for $3$-connected planar graphs with minimum degree $4$ and maximum degree $7$.
  \end{theorem}

\section{\mbox{Other Hardness Results for $3$-Connected Planar  Graphs}}
\label{sec:hardness-other}

In this section we present further strengthenings of hardness results to $3$-connected planar graphs.  The difference to the previous section is that the proofs do not make use of the graph augmentation technique of Lemma~\ref{lem:augmentation}.

By replacing one non-triangular face of size $l$ of a plane graph by the gadget $\Phi_l$, see Fig.~\ref{fig:mis-triangulation-inner}, we can prove the following.

\begin{lemma}
  \label{lem:mis-triangulation-gadget}
  Let~$G$ be a $2$-connected plane graph and $f_{\geq 4}(G)$ be the
  number of faces of $G$ whose size is larger than $3$. There exists a
  $2$-connected plane graph $G'$ such that (i)
  $\alpha(G')=\alpha(G)+1$ and (ii)
  $f_{\geq 4}(G')= f_{\geq 4}(G) -1$.
\end{lemma}

% $\Phi_6$, where the edges of the
%     cycle $C_6$ and of the complete graph on vertices $\{a,b,c,d\}$
%     are depicted with thick lines.}

% we show a reduction from a
% variant of the NP-complete problem {\sc Planar Steiner Tree} ({\sc PST})~\cite{gj-rstpnpc-77},
% defined as follows: Given an instance $\langle G(V,E), S, k \rangle$
% of {\sc PST}, where $G(V,E)$ is a planar graph whose edges have weights
% $\omega:E\rightarrow \mathbb{N}$, $S\subset V$ is a set of
% \emph{terminals}, and $k >0$ is an integer; does a tree $T^*(V^*,E^*)$
% exist such that (1) $V^*\subseteq V$, (2) $E^*\subseteq E$, (3)
% $S\subseteq V^*$, and (4) $\sum_{e\in E^*}\omega(e)\leq k$? The edge
% weights in $\omega$ are bounded by a polynomial function $p(n)$
% (see~\cite{gj-rstpnpc-77}).

It is known that \mis is NP-complete for 2-connected planar
graphs~\cite{mohar-fcgpag-01}.  Iteratively applying the construction
from Lemma~\ref{lem:mis-triangulation-gadget} yields the following.

\begin{theorem}
  \label{thm:mis}
  {\em \mis} is NP-complete for planar triangulations.
\end{theorem}

Next, we show that {\sc Steiner Tree} is NP-complete for 3-connected
cubic graphs.  We first show that {\sc Steiner Tree} is NP-complete
for planar graphs of maximum degree~3 by a reduction from {\sc Steiner
  Tree} in planar graphs, which is known to be
NP-complete\cite{gj-rstpnpc-77}.

\begin{lemma}
  \label{lem:2-connected-steiner-tree}
  {\sc Steiner Tree} is NP-complete for biconnected planar
  graphs of maximum degree~3.
\end{lemma}

\begin{theorem}
  {\sc Steiner Tree} is NP-complete for 3-connected cubic planar graphs.
\end{theorem}

\begin{proof}
  We reduce from {\sc Steiner Tree} in biconnected planar graphs of
  maximum degree~3.  Let $(G,T)$ be such an instance and fix an
  arbitrary planar embedding of $G$.  We call a \emph{chain} a maximal
  path whose internal vertices have degree~2.  After subdividing each
  edge with eight subdivision vertices, we can assume that each chain
  has length at least~$9$.  We now replace each chain of
  length~$\ell$, whose endpoints are degree-3 vertices $u$ and $v$, by
  a copy of the gadget from Fig.~\ref{fig:steiner-gadget}, where the
  size is chosen such we can identify the original chain with the bold
  path in the gadget. Denote by $G'$ the resulting graph together with
  its induced embedding from $G$.  For the terminal set we choose $T'
  = T$.  Observe that $G$ is a subgraph of $G'$ (bold paths in the
  gadgets).  Moreover, since the bold paths in the gadgets are
  shortest paths between the two endpoints, it follows that for any
  pair of vertices in $G$ there is a shortest path in $G'$ that uses
  only vertices in $G$.  Hence $(G,T)$ and $(G',T')$ are equivalent
  instances of {\sc Steiner Tree}.  Observe further that $G'$ is
  cubic.  It remains to make it 3-connected.

  \begin{figure}[tb]
    \centering
    \subfloat[]{\includegraphics[page=1]{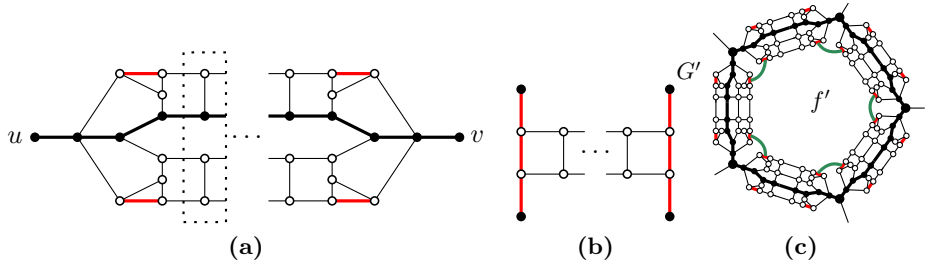}\label{fig:steiner-gadget}}\hfill
    \subfloat[]{\includegraphics[page=2]{fig/steiner}\label{fig:steiner-ladder}}
    \subfloat[]{\includegraphics[page=3]{fig/steiner}\label{fig:steiner-example}}
    \caption{Hardness of {\sc Steiner Tree} for 3-connected cubic graphs.}
    \label{fig:steiner}
  \end{figure}

  % Example of the reduction for {\sc Steiner Tree} for
  %     3-connected cubic graphs, the augmentation operations using the
  %     ladder gadget are indicated by green edges for face $f'$.

  To complete the construction, for each face $f$ of $G$, we traverse
  its boundary and for any two consecutive chains with endpoints $u,v$
  and $v,w$ in such a face, we perform in $G'$ an augmentation
  operation on the two red edges $e,e'$ that are closest to $v$ and
  that are incident to the corresponding face $f'$ in $G'$ using a
  ladder gadget of sufficient length (Fig.~\ref{fig:steiner-ladder})
  so that the ladder gadget cannot be used as a shortcut.  Finally,
  observe that the augmentation operations ensure that the final graph
  is 3-connected and moreover, they do not change the lengths of the
  bold paths; see Fig.~\ref{fig:steiner-example} for an illustration
  of the whole reduction.  The resulting instance $(G',T')$ is hence
  still equivalent to the original instance $(G,T)$.  The reduction
  can clearly be carried out in polynomial time.
\end{proof}

\clearpage

\bibliographystyle{splncs03} \bibliography{augmentation}

\clearpage
\appendix

\section{Preliminaries and Definitions}\label{apx:pre}

\subsection{Drawings and Embeddings} 
A \emph{planar drawing} $\Gamma$ of a graph maps vertices to points in the plane and edges to internally disjoint curves. Drawing $\Gamma$ partitions the plane into topologically connected regions, called  {\em faces}.
The bounded faces are \emph{internal} and the unbounded face is the \emph{outer face}. 
%A vertex (an edge) is {\em external} if it is incident to the outer face and {\em internal} otherwise.
A planar drawing determines a circular ordering of the edges incident to each vertex. Two planar drawings of a connected planar graph are \emph{equivalent} if they determine the same orderings and have the same outer face. A \emph{combinatorial embedding} is an equivalence class of planar drawings.

\subsection{Connectivity and SPQR-trees}\label{apx:SPQR}

A graph $G$ is \emph{connected} if there is a path between any two vertices.
A \emph{cutvertex} is a vertex whose removal disconnects the graph.
A \emph{separating pair} is a pair of vertices $\{u,v\}$ whose removal disconnects the graph.
A connected graph is \emph{$2$-connected} if it does not have a cutvertex and a $2$-connected graph is $3$-connected if it does not have a separating pair.
A $2$-connected plane graph $G$ is {\em internally $3$-connected} if $G$ can be extended to a $3$-connected planar graph by adding a vertex in the outer face and joining it to all the vertices incident to the outer face.

We consider $uv$-graphs with two special \emph{pole} vertices $u$ and $v$, which can be constructed in a fashion very similar to series-parallel graphs.  Namely, an edge $(u,v)$ is an $uv$-graph with poles $u$ and $v$.  Now let $G_i$ be an $uv$-graph with poles $u_i,v_i$ for $i=1,\dots,k$ and let
$H$ be a planar graph with two designated vertices $u$ and $v$ and $k+1$ edges $uv, e_1,\dots,e_k$.
We call $H$ the \emph{skeleton} of the composition and its edges are called \emph{virtual edges}; the edge $uv$ is the \emph{parent edge} and $u$ and $v$ are the poles of the skeleton $H$.
To compose the $G_i$ into an $uv$-graph with poles $u$ and $v$, we remove the edge $uv$ and replace each $e_i$ by $G_i$ for $i=1,\dots,k$ by removing $e_i$ and identifying the poles of $G_i$ with the endpoints of $e_i$.  In fact, we only allow three types of compositions: in a \emph{series composition} the
skeleton $H$ is a cycle of length~$k+1$, in a parallel composition $H$ consists of two vertices connected by $k+1$ parallel edge, and in a \emph{rigid composition} $H$ is 3-connected.

It is known that for every $2$-connected graph $G$ with an edge $uv$ the graph $G-st$ is an $uv$-graph with poles $u$ and $v$.  Much in the same way as series-parallel graphs, the $uv$-graph $G \setminus uv$ gives rise to a (de-)composition tree~$\mathcal T$ describing how it can be obtained from single edges.
The nodes of $\mathcal T$ corresponding to edges, series, parallel, and rigid compositions of the graph are \emph{Q-, S-, P-, and R-nodes}, respectively.  To obtain a composition tree for $G$, we add an additional root Q-node representing the edge $uv$.  To fully describe the composition, we
associate with each node~$\mu$ its skeleton denoted by $\skel(\mu)$.
%For a Q-node~$\mu$, the skeleton consists of the two endpoints of the edge represented by~$\mu$ and one real and one virtual edge between them representing the rest of the graph.
For a node~$\mu$ of $\mathcal T$, the \emph{pertinent graph} $\pert(\mu)$ is the subgraph represented by the subtree with root~$\mu$.
Similarly, for a virtual edge $\eps$ of a skeleton~$\skel(\mu)$, the \emph{expansion graph} of $\eps$, denoted by $\expd(\eps)$ is the pertinent graph $\pert(\mu')$ of the neighbour $\mu'$ of $\mu$ corresponding to $\eps$ when considering $\mathcal T$ rooted at $\mu$.

The \emph{SPQR-tree} of $G$ with respect to the edge $uv$, originally introduced by Di Battista and Tamassia~\cite{dt-ogasp-90}, is the (unique) smallest decomposition tree~$\mathcal T$ for $G$.  Using a different edge $u'v'$ of $G$ and a composition of $G-u'v'$ corresponds to rerooting $\mathcal T$
at the node representing $u'v'$.  It thus makes sense to say that $\mathcal T$ is the SPQR-tree of $G$.  The SPQR-tree of $G$ has size linear in $G$ and can be computed in linear time~\cite{gm-lis-01}.  Planar embeddings of $G$ correspond bijectively to planar embeddings of all skeletons of
$\mathcal T$; the choices are the orderings of the parallel edges in P-nodes and the embeddings of the R-node skeletons, which are unique up to a flip.  When considering rooted SPQR-trees, we assume that the embedding of $G$ is such that the root edge is incident to the outer face, which is
equivalent to the parent edge being incident to the outer face in each skeleton.
We remark that in a planar embedding of $G$, the poles of any node $\mu$ of $\mathcal{T}$ are incident to the outer face of
$\pert(\mu)$. Hence, in the following we only consider embeddings of the pertinent graphs with their poles lying on the same face and refer to such embeddings as {\em regular}. 

Let $\mu$ be a node of $\mathcal T$, we denote the poles of $\mu$ by $u(\mu)$ and $v(\mu)$, respectively.
In the remainder of the paper, we will assume edge $(u(\mu),v(\mu))$ to be part of $\skel(\mu)$ and $\pert(\mu)$.
The outer face of a  (regular) embedding of $\pert(\mu)$ is the one obtained from such an embedding after removing the $(u(\mu),v(\mu))$ connecting its poles.
%\todo{it would be useful to say that we only consider embeddings of the skeleton of a node $\mu$ such that the poles of $\mu$ are incident to the outer face and the edge representing the parent has been removed. In this way it is easy to talk about the outer face of $\mu$ (See for example the construction of the auxiliary graph $O$ in the proof of Lemma~\ref{{lem:general-algo-pnode}}).}
Also, the two paths incident to the outer face of $\pert(\mu)$ between $u(\mu)$ and $v(\mu)$ are called {\em boundary paths} of $\pert(\mu)$.

\subsection{Canonical Ordering}

Let $G = (V, E)$ be a $3$-connected plane graph with vertices $v_2$, $v_1$, and $v_n$ in this clockwise order along the outer face of $G$. Let $\pi = (P_1,\ldots,P_k)$ be an ordered partition of $V$ into paths, where $P_1=(v_1,v_2)$ and $P_k=(v_n)$. Define $G_i$ to be the subgraph of
$G$ induced by $P_1\cup \ldots \cup P_i$, and denote by $C_i$ the boundary of the outer face of $G_i$. We say that $\pi$ is a canonical ordering for $G$ if:
\begin{itemize}
  % \item $P_1$ consists of $\{v_1, v_2\}$, where $v_2$ lies on the outer face and $(v_1, v_2) \in E$.
  % \item $V_k$ is a singleton $\{v_n\}$, where $v_n$ lies on the outer face, $(v_1, v_n) \in E$, and $v_n \neq v_2$.
\item each $C_i (i > 1)$ is a cycle containing $(v_1, v_2)$.
\item each $G_i$ is biconnected and internally $3$-connected, that is, removing two interior vertices of $G_i$ does not disconnect it; and
\item for each $i \in \{2,\ldots, k-1\}$, one of the two following conditions holds:
  \begin{itemize}
  \item (a) $P_i$ is a singleton, $\{z\}$, where $z$ belongs to $C_i$ and has at least one neighbor in $G \setminus G_i$.
  \item (b) $P_i$ is a chain, $\{z_1,\ldots,z_l\}$, where each $z_j$ has at least one neighbour in $G \setminus G_i$, and where $z_1$ and $z_l$ each have one neighbour on $C_{i-1}$, and these are the only two neighbors of $P_i$ in $G_{i-1}$.
  \end{itemize}
\end{itemize}

Observe that, by condition (b), if $P_i$ is a chain, then the two neighbours of $z_1$ and $z_l$ on $C_{i-1}$ are adjacent in $G_{i-1}$.

\section{Bounded-degree Augmentation} \label{apx:augmentation}

In this section we give an algorithm (Lemma~\ref{lem:augmentation}) to augment a $2$-connected planar graph $G$ with minimum degree $\delta(G)\geq 3$ and maximum degree $\Delta(G)$ to a $3$-connected planar graph $G'$ with $\delta(G')=\delta(G)$ and $\Delta(G')=\Delta(G)$ by applying
  $H$-splits to disjoint pairs of edges of $G$.

\begin{figure}[tb!]
  \centering \subfloat[]{ \includegraphics[page=1]{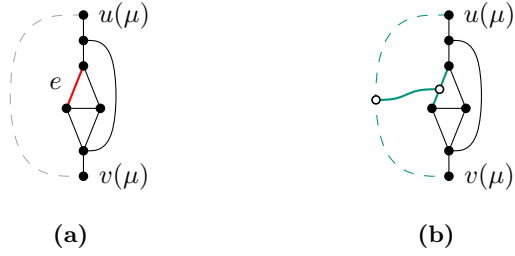} } \hfil \subfloat[]{ \includegraphics[page=2]{fig/ext-triconnectible} }
  \caption{(a) Pertinent graph of an $e$-externally $3$-connectible S-node $\mu$ and (b) the subdivision of a $3$-connected planar graph obtained by performing the $H$-split on $\langle e, (u(\mu),v(\mu))\rangle$ in graph $\pert(\mu)$.}
  \label{fig:ext-triconnectible}
\end{figure}

We start with some definitions.  Let $e$ be an edge of $\pert(\mu)$ incident the outer face of some regular embedding $\mathcal E_\mu$ of $\pert(\mu)$. Observe that, by the definition of outer face of a pertinent graph, it holds that $e \neq (u(\mu),v(\mu))$.  Then, $\mathcal E_{\mu}$ is {\em
  $e$-externally $3$-connectible} if either
\begin{inparaenum}
\item $\mu$ is a Q-node, that is, $\pert(\mu)=e$, or
\item the graph obtained from the $\pert(\mu)$ by performing an $H$-split on $\langle e, (u(\mu),v(\mu)))\rangle$ is a subdivision of a $3$-connected graph whose unique subdivision vertices are the poles of $\mu$; refer to Fig.~\ref{fig:ext-triconnectible}.
\end{inparaenum}
Also, we say that $\mu$ is {\em $e$-externally $3$-connectible} (or, simply, {\em externally $3$-connectible}) if $\pert(\mu)$ admits an $e$-externally $3$-connectible embedding, for some edge $e$ of $\pert(\mu)$.

%%%%%%%%%%%%%%%%% REMOVED PROPERTIES%%%%%%%%%%%%%%%%%%%%%%%%%%%%
\remove{ Let $\mu$ be an $e$-externally $3$-connectible node for some edge $e$ of $\pert(\mu)$.
  % We say that $\mu$ is of {\em type A}, {\em B}, or {\em C} if $\mu$ is an R-, an S-, or a P-node respectively; see Figs....., respectively.

  The following two structural properties hold when $\mu$ is a P- or an S-node; refer to Figs.~\ref{fig:ext-types-S} and~\ref{fig:ext-types-P}, respectively.

\begin{figure}[tb!]
  \centering \subfloat[]{ \includegraphics[height=0.25\textwidth,page=2]{fig/ext-triconnectible-types}\label{fig:ext-types-S} } \hfil \subfloat[]{ \includegraphics[height=0.25\textwidth,page=1]{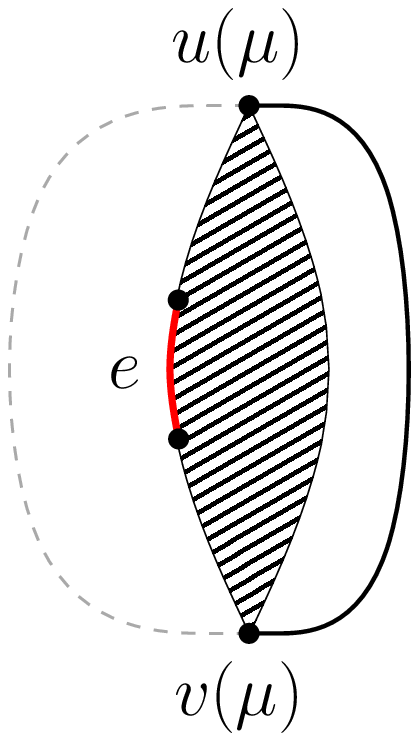}\label{fig:ext-types-P} }
  \caption{Illustration of the structure of the skeleton of an $e$-externally $3$-connectible node $\mu$ in the case in which $\mu$ is an (a) S-node or a (b) P-node.  Dotted virtual edges represent optional Q-nodes, while stripped thick edges represent the $e$-externally $3$-connectible children.}
  \label{fig:ext-types}
\end{figure}

\begin{property}\label{prop:S-nodes}
  If $\mu$ is an $e$-externally $3$-connectible S-node, then:
  \begin{itemize}
  \item $\mu$ has at most three children,
  \item there exists exactly a non-Q-node child of $\mu$ and this child is an $e$-externally $3$-connectible R- or P-node, and
  \item the virtual edges representing the (at most two) Q-node children of $\mu$ are incident to the poles of $\mu$ in $\skel(\mu)$.
  \end{itemize}
\end{property}

Property~\ref{prop:S-nodes} descends from the fact that since $\delta(G)\geq 3$ and since each internal node of $\cal T$ has at least two children, then at least one child $\nu$ of $\mu$ has to be a non-Q-node child and no two Q-node children can be adjacent in $\skel(\mu)$, as otherwise there would
exists a cut-vertex of $\skel(\mu)$ incident to two virtual edges representing adjacent Q-nodes, which is also a degree-$2$ vertex in $G$, contradicting the fact that $\delta(G)\geq 3$.  Furthermore, if there were two or more non-Q-node children, then performing a single $H$-split won't suffice to
make $\pert(\mu)$ $3$-connected.  Finally, if the unique non-Q-node child of $\mu$ were not $e$-externally $3$-connectible, neither will be $\mu$.

\begin{property}\label{prop:P-nodes}
  If $\mu$ is an $e$-externally $3$-connectible P-node, then:
  \begin{itemize}
  \item $\mu$ has exactly two children, one of which is a Q-node, and
  \item the unique non Q-node child of $\mu$ is an $e$-externally $3$-connectible R- or S-node.
  \end{itemize}
\end{property}

Property~\ref{prop:P-nodes} descends from the fact each internal node of $\cal T$ has at least two children and that, since $\mu$ is a P-node, at most one of its children is a Q-node. Hence, if there were two or more non-Q-node children of $\mu$, then performing a single $H$-split won't suffice to
make $\pert(\mu)$ $3$-connected. Again, if the unique non-Q-node child of $\mu$ were not $e$-externally $3$-connectible, neither will be $\mu$.  }
%%%%%%%%%%%%%%%%%%%%%%%%%%%%%%%%%%%%%%%%%%%%%%%%%%%%%%%%%%%%%%

Let $\mu$ be a node of $\mathcal T$, let $e$ be an edge of $\pert(\mu)$, and let $\pert^*(\mu)$ be a graph obtained by applying $H$-splits on distinct pairs of edges in $\pert(\mu)$. We say that $e$ is a {\em free edge} if $e \in E(\pert(\mu)) \setminus E(\pert^*(\mu))$, that is, edge $e$ has not
been used in any $H$-split.

Consider a non-Q-node $\mu$. Let $\pert^*(\mu)$ be a graph obtained from $\pert(\mu)$ via a set of edge-disjoint $H$-splits, let $L_\mu=[e_1,e_2]$ and $R_\mu=[e_3]$ be two lists of free edges in $E(\pert^*(\mu))$, and let $\mathcal E^*_\mu$ be a regular embedding of $\pert^*(\mu)$. We say that the $4$-tuple $\langle \pert^*(\mu), \mathcal E^*_\mu, L_\mu, R_\mu \rangle $ is {\em extendible} if $L_\mu$, and $R_\mu$ are incident to different faces of $\mathcal E^*_\mu$ incident to the parent edge and $\mathcal E^*_\mu$ is $e$-externally $3$-connectible with $e \in L_\mu$.  
Observe that,
once $\pert^*(\mu)$, $L_\mu$, and $R_\mu$ have been fixed, there exists a unique (up to a flip) embedding $\mathcal E^*_\mu$ of $\pert^*(\mu)$ such that $\langle \pert^*(\mu), \mathcal E^*_\mu, L_\mu, R_\mu \rangle $ is extendible.  Hence, to easy the notation, in the following we will omit to
specify the embedding of $\pert^*(\mu)$. %Also, we say that $\mu$ is {\em extendible} if $\langle \pert^*(\mu), L_\mu, R_\mu \rangle $ is.
If $\mu$ is a Q-node representing edge $e=(u(\mu),v(\mu))$, then we also say that the triple $\langle e, L_\mu=[e], R_\mu=[e] \rangle $ is {\em extendible}, thus allowing $|L_\mu|=1$ and $L_\mu\cap R_\mu \neq \emptyset$ in this case.

For simplicity, we will also use the notation $\expd^*(e_\mu)$ to refer to $\pert^*(\mu)$ where $e_\mu$ is the virtual edge representing $\mu$ in the skeleton of its parent.

Let $\mu$ be an internal node in $\mathcal T$ with children $\mu_1,\ldots,\mu_k$ and let $f$ be a face of an embedding $\mathcal E_\mu$ of $\skel(\mu)$. Consider the clockwise sequence of virtual edges in $f$. This sequence induces a natural order of the lists $L_{\mu_i}$ of the children of $\mu$
whose corresponding virtual edges bound $f$ such that performing $H$-splits between consecutive pairs of free edges of such children does not violate planarity. Hence, in the following we will always assume lists $L_{\mu_i}$ to be ordered according to such a natural order. %\todo{Not sure this is the best way to say this.}

We root the SPQR-tree $\cal T$ of $G$ to an arbitrary Q-node $\rho$ whose unique child $\xi$ is an R-node. Observe that such a Q-node exists since $\delta(G)\geq 3$.  We process the nodes of $\mathcal T$ bottom-up and show how to compute for each node $\mu$ with children $\mu_1,\ldots,\mu_k$ an
extendible triple $\langle \pert^*(\mu), L_\mu, R_\mu \rangle $, starting from the extendible triples $\langle \pert^*(\mu_i), L_{\mu_i}, R_{\mu_i} \rangle $ of its children.
When we reach the root $\rho$, performing the $H$-split on $\langle (u(\rho),v(\rho)), e \in L_\xi \rangle$ in the graph $\pert^*(\xi)$, clearly yields a $3$-connected planar graph with the desired properties, as the triple $\langle \pert^*(\xi), L_\xi, R_\xi \rangle$ is extendible.

We show how to compute extendible triples for each non-root node~$\mu \in \cal T$.

Suppose $\mu$ is a Q-node representing edge $e=(u(\mu),v(\mu))$. In this case there is nothing to be done as the triple $\langle pert^*(\mu)=e,L_\mu=[e],R_\mu=[e]\rangle$ is extendible by definition. 

Suppose $\mu$ is an S-node with children $\mu_1,\ldots,\mu_k$.  Recall that since $\mu$ is an internal node of $\cal T$, it has at least two children; also, since $\delta(G) \geq 3$, no two Q-nodes are adjacent in $\skel(\mu)$. Hence, $\mu$ has at least a non-Q-node child $\nu$.
We distinguish two cases based on whether the number of non-Q-node children of $\mu$ is larger than one or not.

\begin{figure}[tb!]
\centering
  \subfloat[]{\includegraphics[page=1]{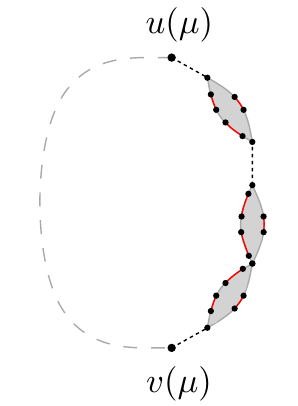}\label{fig:ext-S-node-a}}
  \subfloat[]{\includegraphics[page=2]{fig/ext-S-node}\label{fig:ext-S-node-b}}
  \subfloat[]{\includegraphics[page=2]{fig/ext-P-node}\label{fig:ext-P-node-a}}
  \subfloat[]{\includegraphics[page=2]{fig/ext-P-node}\label{fig:ext-P-node-b}}
\caption{Illustration for the proof of Lemma~\ref{lem:augmentation} when $\mu$ is ((a),(b)) an S-node with three non-Q-node children $\nu_1$, $\nu_2$, and $\nu_3$, and ((c),(d)) a P-node with a Q-node child.
  ((a),(c)) Auxiliary graph $\overline{\pert(\mu)}$.
  ((b),(d)) Augmentation of $\overline{\pert(\mu)}$ to $\pert^*(\mu)$ via $H$-splits.}
\label{fig:ext-SP-node} 
\end{figure}

{\bf Case S1.} Assume $\nu$ is the unique non-Q-node child of $\mu$ and let $e_\nu$ be the corresponding virtual edge in $\skel(\mu)$.  We construct an extendible triple for $\mu$ as follows.  Let $\overline{\pert(\mu)}$ be the graph obtained from $\skel(\mu)$ by replacing $e_{\nu}$ in $\skel(\mu)$
with $\expd^*(e_{\nu})$. We set $\pert(\mu^*)=\overline{\pert(\mu)}$, $L_\mu = L_\nu$, and $R_\mu = R_\nu$.  The fact that the constructed triple is extendible is due to the fact that (i) $\langle \pert^*(\nu), L_\nu,R_\nu \rangle$ is extendible and that (ii) since $\nu$ is the unique non-Q-node
child of $\mu$ and since the remaining (at most two) Q-node children of $\mu$ are not adjacent in $\skel(\mu)$, the poles of the corresponding virtual edges cannot be part of a separation pair of $\pert^*(\mu)$.

{\bf Case S2.} Let $\nu_1,\ldots,\nu_s$ be the non-Q-node children of $\mu$ ordered as the corresponding virtual edges appear in $\skel(\mu)$ from $u(\mu)$ to $v(\mu)$.  We construct an extendible triple for $\mu$ as follows; refer to Fig.~\ref{fig:ext-SP-node}.
First, we construct an auxiliary graph $\overline{\pert(\mu)}$ starting from $\skel(\mu)$ by replacing each virtual edge $e_{\nu_i}$ in $\skel(\mu)$ corresponding to a non-Q-node child $\nu_i$ of $\mu$ with the expansion graph $\expd^*(e_{\nu_i})$ of $\nu_i$, for $i=1, \ldots, s$; see Fig.~\ref{fig:ext-S-node-a}.
Then, we obtain $\pert^*(\mu)$ from $\overline{\pert(\mu)}$ by performing an $H$-split on $\langle L_{\nu_i}[2], L_{\nu_{i+1}}[1]\rangle$, for $i=1,\ldots,s-1$; see Fig.~\ref{fig:ext-S-node-b}.
Finally, we set $L_\mu=[L_{\nu_1}[1],L_{\nu_s}[2]]$ and $R_\mu=R_{\nu_1}$.
We now show that the constructed triple is extendible.
Observe that each $H$-split can be seen as an operation that turns two non-Q-node children of $\mu$, together with the unique Q-node possibly separating them in $\skel(\mu)$, into a single externally $3$-connectible child of $\mu$.
Hence, at the end of the augmentation, there might exist at most two non-adjacent virtual edges representing Q-node children of $\mu$ left in $\skel(\mu)$, whose poles do not contribute to any separation pair of $\pert^*(\mu)$.

Suppose $\mu$ is an P-node. Recall that at most one child of $\mu$ can be a Q-node.
We construct an extendible triple for $\mu$ as follows; refer to Fig.~\ref{fig:ext-SP-node}.
First, we select an arbitrary embedding of $\skel(\mu)$ such that the unique child of $\mu$ that is a Q-node, if any, is incident to the outer face.
Let ${\nu_1},\ldots,{\nu_{s}}$ be the clockwise ordering of the non-Q-node children of $\mu$ around $u(\mu)$ determined by such an embedding, where ${\nu_1}$ is a non-Q-node child of $\mu$ incident to the outer face.
Second, we construct an auxiliary graph $\overline{\pert(\mu)}$ staring from $\skel(\mu)$ by replacing each virtual edge $e_{\nu_i}$ in $\skel(\mu)$ with $\expd^*(e_{\nu_i})$, for $i=1,\ldots,s$; see Fig.~\ref{fig:ext-P-node-a}.
Third, we obtain $\pert^*(\mu)$ from $\overline{\pert(\mu)}$  by performing an $H$-split on $\langle R_{\nu_i}[1], L_{\nu_{i+1}}[2] \rangle$, for $i=1,\ldots,s-1$; see Fig. ~\ref{fig:ext-P-node-b}.
Finally, we set $L_\mu=L_{\nu_1}$, and $R_\mu=R_{\nu_{s}}$, if there exists no Q-node child of $\mu$, or $R_\mu=[(u(\mu),v(\mu))]$, otherwise.

Suppose $\mu$ is an R-node.  Let $\mathcal{E}_\mu$ be the unique (up to a flip) regular embedding of $\skel(\mu)$ and let $\pi=(P_1=(u(\mu),v_2),P_2,\ldots,P_k=v(\mu))$ be a canonical ordering of $\skel(\mu)$ where $v_2$ is the neighbour of $u(\mu)$ different from $v(\mu)$ that is incident to the
outer face of $\mathcal{E}_\mu$.
Also, let $skel_i(\mu)$ be the embedded subgraph of $skel(\mu)$ induced by vertices $\bigcup^i_{j=1} P_j$ and $C_i$ be the cycle bounding the outer face of $skel_i(\mu)$. Further, let $\overline{\pert_i(\mu)}$ be the graph obtained from  $skel_i(\mu)$ as illustrated before.
% Finally, let $\overline{\pert_i(\mu)}$ be the graph obtained from $\skel_i(\mu)$ by replacing each virtual edge $e_{\nu_i}$ in $\skel_i(\mu)$ corresponding to a non-Q-node child $\nu_i$ of $\mu$ with $\expd^*(e_{\nu_i})$.

We show how to augment $\overline{\pert_i(\mu)}$ to a new graph $\pert^*_i(\mu)$ via $H$-splits so that $\pert^*_i(\mu)$ is internally $3$-connected and the outer face of $\pert^*_i(\mu)$ contains at least $|C_i|-1$ free edges, each of which is separated by two vertices of $\skel_i(\mu)$.
Hence, $\pert^*_k(\mu)$ contains free edges on its outer face that can be used to instantiate $L_\mu$ and $R_\mu$.  Further, since $\skel_k(\mu)=\skel(\mu)$ is $3$-connected, so is $\pert^*_k(\mu)=\pert^*(\mu)$.
Hence, $\pert^*(\mu)$ is trivially $e$-externally $3$-connectible with respect to any edge $e \in L_\mu$.

The augmentation in done by induction on $i$.
%
% The base case is $i=2$. In this case, $skel_2(\mu)$ is the cycle induced by $V_1 = \{u(\mu),v_2\} \cup V_2$. Denote by $e_{\nu_1},\ldots,e_{\nu_c}$ the virtual edges of $skel_2(\mu)$ in the clockwise order in which they appear along $\skel_2(\mu)$ starting from $v(\mu)$, where
% $e_{\nu_1}=(u(\mu),v(\mu))$.  For $i=1,\ldots,c-1$, we obtain $\pert^*_i(\mu)$ by performing an $H$-split in $\overline{\pert_i(\mu)}$ on $\langle L_{\nu_i}[2], L_{\nu_{i+1}}[1]\rangle$; see Fig. ~\ref{fig:ext-types-R-a}. Observe that, for each node $\nu_i$, the edge in $R_{\nu_i}$ has not been
% used by any $H$-split and is incident to the outer face of $\pert^*_i(\mu)$, hence such a face contains $C_i-1$ free edges, each of which is separated by two vertices of $skel_i(\mu)$. Also, $\pert^*_i(\mu)$ is internally triconnected since $\skel_i(\mu)$ is internally triconnected and since once
% an $H$-split is performed the poles of the virtual edge of $\skel_i(\mu)$ interested by the $H$-split are not a separation pair in $\pert^*_i(\mu)$.
We first assume that the child of $\mu$ corresponding to edge $\ell = (u(\mu),v_2)$ is a Q-node; we will show how to remove such an assumption at the end of the construction by performing a special~{$H$-split}. 

The base case is $i=1$. In this case, $skel_i(\mu)$ consists of the single virtual edge $\ell$. We obtain $\pert^*_i(\mu)$ by simply replacing edge $\ell$ with $\expd^*(\ell)$. Observe that, the outer face of $\pert^*_i(\mu)$ contains at least a free edge (in fact, exactly one since the
child of $\mu$ corresponding to $e$ is a Q-node). Also, $\pert^*_i(\mu)$ is internally $3$-connected, since $\expd^*(\ell)$  is externally $3$-connectible.

In the inductive step $1 < i \leq k$, assume we have already computed graph $\pert^*_{i-1}(\mu)$. By the inductive hypothesis, $\pert^*_{i-1}(\mu)$ is internally $3$-connected and the outer face of $\pert_{i-1}^*(\mu)$ contains at least $|C_{i-1}|-1$ free edges, each of which is separated by two vertices of $\skel_{i-1}(\mu)$.
Let $P_i=(l_1,\ldots,l_c)$ be the $i$-th element in $\pi$. Suppose $|P_i| > 1$, that is, $P_i$ is not a singleton vertex; the case $|P_i| =1$ being simpler.  Observe that, since $\pi$ is a canonical ordering, there exists a virtual edge $(l_1,l_k)$ that is incident to the outer face of
$\skel_{i-1}(\mu)$.
Let $e$ be the free edge along the outer face of $\pert^*_{i-1}(\mu)$ between $l_1$ and $l_k$ and belonging to the expansion graph of virtual edge $(l_1,l_k)$.  We construct $\pert_{i}^*(\mu)$ as follows.
First, initialise $\pert_{i}^*(\mu)$ to $\pert_{i-1}^*(\mu)$.
Let $P^*_i=(x=l_0,l_1,\ldots,l_c,l_{c+1}=y)$ be the list of vertices obtained by prepending and appending to $P_i$ vertices $x$ and $y$, respectively, where $x$ and $y$ are the only two neighbours of $l_1$ and $l_c$ in $\pert_{i}^*(\mu)$, respectively.  Add to $\pert_{i}^*(\mu)$ vertices
$l_1,\ldots,l_c$ in the outer face of $\pert_{i}^*(\mu)$ and edges $(l_j,l_{j+1})$, with $0 \leq j < c$.
Denote by $\nu_j$ the child of $\mu$ corresponding to the virtual edge $e_j$, for $j=1,\dot,c$, and
let $H_i$ be the resulting graph; see Fig.~\ref{fig:ext-R-node-a}.
Then, replace each edge $e_j=(l_j,l_{j+1})$ with $\expd^*(e_j)$, for $j=0,\ldots, c$ (except, when $i=k$, for the edge
representing the parent of $\mu$).
Let $\overline{H_i}$ be the resulting graph; see Fig.~\ref{fig:ext-R-node-b}.
Finally, we obtain $\pert^*_i(\mu)$ from $\overline{H_i}$ by performing an $H$-split (i) on 
$\langle e, L_{\nu_1}[1],\rangle$ and (ii) on $\langle L_{\nu_j}[2], L_{\nu_{j+1}}[1]\rangle$, for $j=1,\ldots,c-1$; see Fig.~\ref{fig:ext-R-node-c}.
Graph $\pert_{i}^*(\mu)$ is internally $3$-connected, since $\skel_i(\mu)$ is internally $3$-connected and since once an
$H$-split is performed the poles of the virtual edge $e_{\nu_j}$ of $\skel_i(\mu)$ whose expansion graph $\expd^*(e_{\nu_j})$ interested by the $H$-split do not belong to a separation pair in $\pert^*_i(\mu)$.
Further, there exists a free edge on the outer face of $\pert_{i}^*(\mu)$ between each pair of consecutive vertices in $(l_0,l_1,\ldots,l_c,l_{c+1})$, namely, for each virtual edge
$e_j=(l_j,l_{j+1})$ with $0 \leq j \leq c$, edge $R(\nu_j)$ is incident to the outer face of $\pert_{i}^*(\mu)$.

\begin{figure}[tb!]
  \centering
  \subfloat[]{ \includegraphics[height=0.16\textwidth,page=1]{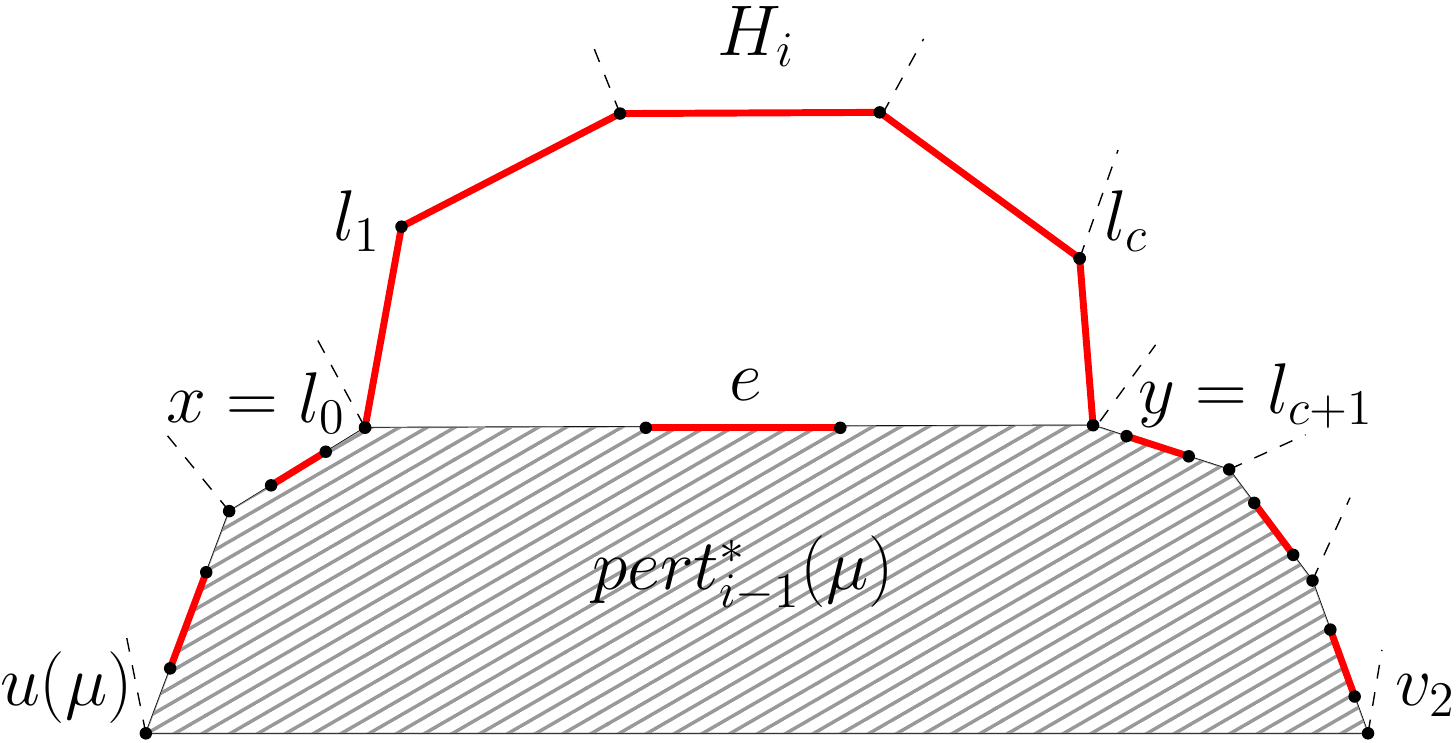}\label{fig:ext-R-node-a} }
  \hfil
  \subfloat[]{ \includegraphics[height=0.16\textwidth,page=2]{fig/ext-R-node}\label{fig:ext-R-node-b} }
  \hfil
  \subfloat[]{ \includegraphics[height=0.16\textwidth,page=3]{fig/ext-R-node}\label{fig:ext-R-node-c} }
  \caption{Illustration for the proof of Lemma~\ref{lem:augmentation} when $\mu$ is an R-node.
    Graphs~(a)~$H_i$ and~(b)~$\overline{H_i}$.
    (c) Augmentation of $\overline{H_i}$ to $\pert_i^*(\mu)$ via $H$-splits.}
  \label{fig:ext-R-node}
\end{figure}

To complete the proof, we only need to show that the child of $\mu$ corresponding to the virtual edge $\ell=(u(\mu),v_2)$ needs not to be a Q-node.
Observe that, in the inductive construction we didn't make use of the free edge corresponding to edge $\ell$ in any $H$-split.
Let $\nu_\ell$ be the child of $\mu$ corresponding to $\ell$ and let $\tau$ be the child of $\mu$ whose corresponding virtual edge is incident to $\ell$ and to the outer face of $\skel(\mu)$.  Once $\pert_{k}^*(\mu)$ has been constructed, we replace $\ell$ with $\expd^*(\ell)$ and perform an $H$-split on $\langle L_{\nu_\ell}[2], L_{\tau}[1] \rangle$.

Altogether we have proved the following main result.

\rephrase{Lemma}{\ref{lem:augmentation}}{
  Let~$G=(V,E)$ be a $2$-connected planar graph with minimum degree $\delta(G)\geq 3$ and maximum degree $\Delta(G)$.  There exist disjoint pairs $(e'_1,e''_1)$, $\ldots$, $(e'_k,e''_k)$ of edges in $E$ such that performing the $H$-splits
  $\langle e'_1,e''_1\rangle$, $\ldots$, $\langle e'_k,e''_k\rangle$ yields a $3$-connected planar graph $G'$ with $\delta(G')=\delta(G)$ and $\Delta(G')=\Delta(G)$.}

% In the next section we exploit Lemma~\ref{lem:augmentation} to prove NP-hardness of the {\sc Maximum Independent Set} problem for $3$-connected cubic planar graphs.
% In the Appendix\todo{Add a reference to the right appendix}, we will give further examples of application of
%Lemma~\ref{lem:augmentation} to extend NP-hardness results from the class of $2$-connected planar graphs with bounded degree to that of $3$-connected planar graphs with bounded degree.

\section{Omitted Proofs from Section~\ref{sec:hardness-augmentation}}

\rephrase{Theorem}{\ref{thm:mis-cubic}}{
  \mis is NP-complete for $3$-connected cubic planar graphs.}

\begin{proof}
  Let~$G$ be a $2$-connected cubic planar graph. By Lemma~\ref{lem:augmentation}, a $3$-connected cubic planar graph $G'$ can be obtained from $G$ by applying $H$-splits on $k \in O(n)$ distinct pairs of edges in $E(G)$.
  % Observe that, since each edge in $E(G)$ belongs to at most one of
  % such pairs and since, by the planarity of $G$, it holds that
  % $|E(G)| \in O(n)$, we have that $k \in O(n)$.
  
  Let $\langle (a_1,b_1),\ldots,(a_k,b_k)\rangle$ be any ordering of the set of edge pairs of $E(G)$ determined by the algorithm described in the proof of Lemma~\ref{lem:augmentation}.
  We augment graph $G$ to an auxiliary graph $G^*$ as follows. For $i=1,\ldots,k$, we apply the construction illustrated in Fig.~\ref{fig:mis-cubic-gadget} to the edge pair $(a_i,b_i)$
  %\footnote{\red{Here is a minor issue. In our augmentation edges are somehow directed. If ones reverse the direction of exactly one edge then the augmentation might result in a non-planar graph. Let's find a simple way to make it clear!}}.
  It is easy to see that $G^*$ is planar, cubic, and $3$-connected.  Furthermore, graph $G^*$ can be inductively defined as follows. Let $G_0=G$ and $G_{i}$ be the graph
  obtained by applying the construction illustrated in Fig.~\ref{fig:mis-cubic-gadget} to the edge pair $(a_i,b_i)$ of graph $G_{i-1}$. Then, we have that $G^*=G_{k}$.
  
  By Lemma~\ref{lem:independent-increase}, it holds that $\alpha(G_i) = \alpha(G_{i-1}) + 5$, with $i=1,\ldots,k$.  It follows that~$G$ admits an independent set of size~$w$ if and only if~$G^*=G_{k}$ admits an independent set of size~$w + 5k$. Also, $|V(G^*)|=|V(G)|+12 k$. Since $k \in O(n)$ and
  since each graph $G_i$ can be obtained by $G_{i-1}$ in constant time, this is a polynomial-time reduction from \mis in $2$-connected cubic planar graphs to \mis in $3$-connected cubic planar graphs.  The fact that the former is NP-complete~\cite{mohar-fcgpag-01} implies the claim.
\end{proof}

\rephrase{Lemma}{\ref{lem:coloring-strenghten}}{
  The {\sc $3$-Coloring} problem is NP-complete for $2$-connected planar graphs with minimum degree $4$ and maximum degree $7$.
}

\begin{proof}
We show a reduction from the  {\sc $3$-Coloring} problem for $4$-regular planar graphs~\cite{DAILEY1980289}.

    \begin{figure}[tb!]
    \centering
    \subfloat[]{
      \includegraphics[page=2]{fig/cutvertex-coloring}
      \label{fig:cutvertex-coloring-a}
    }\hfil
    \subfloat[]{
      \includegraphics[page=1]{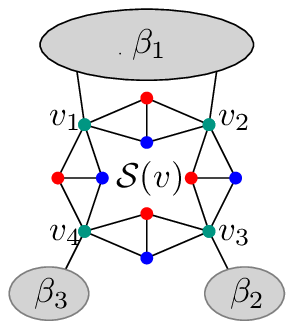}
      \label{fig:cutvertex-coloring-b}
    }
    \caption{
      (a) A cut vertex $v$ incident to three blocks $\beta_1$, $\beta_2$, and $\beta_3$.
      (b) Gadget $\mathcal S(v)$ that replaces $v$ with new cut vertices each incident to exactly two blocks.
    }\label{fig:cutvertex-coloring}
  \end{figure}
  
  Let $G$ be a $4$-regular plane graph $G$.
  First, we replace each cut vertex $v$ incident to more than two blocks with the gadget $\mathcal S(v)$ illustrated in Fig.~\ref{fig:cutvertex-coloring-b} to obtain a new plane graph $G'$. Observe that, $G'$ is $3$-colorable if and only if $G$ is and that each cut vertex
  of $G'$ is now incident to exactly two blocks. Further,  it is easy to verify that $\Delta(G') \leq 5$.

  Second, let $(v,x)$ and $(v,y)$ be two edges of $G'$ appearing consecutively around a cut vertex $v$ and each belonging to a different block. We augment $G'$ to a minimum degree $4$ and maximum degree $6$ planar graph $G''$ such that $G'$ is $3$-colorable if and only if $G''$ is and
  $\mathcal C(G'')=\mathcal C(G')-1$ by replacing the pair $\langle (v,x), (v,y) \rangle$ with $A_1(\langle (v,x), (v,y) \rangle)$ (see Fig.~\ref{fig:A1}).
  Repeating such an augmentation for each cut vertex, eventually yields a $2$-connected planar graph $G^*$ that is $3$-colorable if and only if $G$ is.
  Further, since each augmentation increases the degree of the end points of the selected edge pair by at most $1$ and since each edge might be involved in at most two augmentations (if both its end points are
  cut vertices), it follows that $\Delta(G^*) \leq \Delta(G') + 2 = 7$.
  \end{proof}

  \section{Omitted Proofs from Section~\ref{sec:hardness-other}}

\rephrase{Lemma}{\ref{lem:mis-triangulation-gadget}}{
  Let~$G$ be a $2$-connected plane graph and $f_{\geq 4}(G)$ be the
  number of faces of $G$ whose size is larger than $3$. There exists a
  $2$-connected plane graph $G'$ such that (i)
  $\alpha(G')=\alpha(G)+1$ and (ii)
  $f_{\geq 4}(G')= f_{\geq 4}(G) -1$.
}

\begin{proof}
  Let $f$ be any non-triangular face of $G$ and let $\ell(f)$ be the
  length of $f$. Also, for any $k\geq 4$, let $\Phi_{k}$ be the graph
  constructed as follows. First, intialize $\Phi_{k}$ to the union of
  a cycle $C_k=c_1,c_2,\ldots,c_{k}$ of lenght $k$ and of the complete
  graph on the four vertices $\{a,b,c,d\}$. Then, add to $\Phi_{k}$
  edges $(d,c_i)$, for $i=1,2,\ldots,k$, and edges
  $(c_1,d),(c_1,b),(c_1,c)$ and $(c_2,c)$. Observe that, $\Phi_{k}$ is
  internally triangulated and hence $2$-connected.  Refer to
  Fig.~\ref{fig:mis-triangulation-inner} for an illustration of the
  gadget $\Phi_{6}$.  Graph $G'$ can be obtained from $G$ by
  identifying $f$ with the outer face of gadget $\Phi_{\ell(f)}$.
  Clearly, graph $G'$ is $2$-connected as $G$ and $\Phi_{\ell(f)}$
  are.  Also, since face $f$ has been replaced by an internally
  triangulated graph, it is~$f_{\ge 4}(G') = f_{\ge 4}(G) - 1$.
  
  We now prove that $\alpha(G')=\alpha(G)+1$.  First, observe that an
  independent set~$I'$ in~$G'$ can contain at most one vertex
  in~$\{a,b,c,d\}$, since these vertices induce a~$K_4$.  Hence, $I'
  \cap V(G)$ is an independent set of size~$\alpha(G') - 1$, showing
  that $\alpha(G') \le \alpha(G) + 1$.  Conversely, if~$I$ is an
  independent set in~$G$, then~$I' \cup \{a\}$ is an independent set
  in~$G'$, showing that~$\alpha(G') \ge \alpha(G) + 1$.
\end{proof}

\rephrase{Theorem}{\ref{thm:mis}}{
  {\em \mis} is NP-complete for planar triangulations.
}

\begin{proof}
  Let~$G$ be a $2$-connected plane graph.  Observe that applying the
  reduction from Lemma~\ref{lem:mis-triangulation-gadget} to a face
  of~$G$ whose size is larger than $3$ yields a $2$-connected plane
  graph~$G'$ where $\alpha(G') = \alpha(G)+1$ and the number of
  non-triangular faces of $G'$ is one less the number of
  non-triangular faces of $G$.  Iterating this reduction eventually
  leads to a planar triangulation~$G^*$
  with~$\alpha(G^*) = \alpha(G) + f_{\ge 4}(G)$, where~$f_{\ge 4}(G)$
  denotes the number of faces of size at least four in~$G$.
  Hence,~$G$ admits an independent set of size~$k$ if and only
  if~$G^*$ admits an independent set of size~$k+f_{\ge 4}(G)$.
	
  It follows that the described procedure is a polynomial-time
  reduction from \mis in $2$-connected planar graphs to \mis in planar
  triangulations.  The fact that the former is
  NP-complete~\cite{mohar-fcgpag-01} implies the claim.
\end{proof}

\rephrase{Lemma}{\ref{lem:2-connected-steiner-tree}}{
  {\sc Steiner Tree} is NP-complete for biconnected planar
  graphs of maximum degree~3.
}

\begin{proof}
  We start with an instance $(G,T)$ of (unweighted) {\sc Steiner
    Tree}, where $G$ is a biconnected planar graph with $n$ vertices
  and $m$ edges, and $T \subseteq V(G)$ is a set of terminals.  This
  problem is known to be NP-complete~\cite{gj-rstpnpc-77}.  We first
  construct a new instance $(G',T')$ of {\sc Weighted Steiner Tree}
  where $G'$ is biconnected and has maximum degree-3 by replacing each
  vertex $v$ of degree more than~3 by a cycle $C_v$ of the same length
  so that $G'$ is planar.  The edges of the cycles have weight~1, and
  we give the remaining edges a weight of $2m+1$.  For each $v \in
  V(G)$, if $v$ has degree at most~3, it is $v \in T'$, otherwise, we
  choose an arbitrary vertex $u \in C_T$ and put it in $T'$.

  A Steiner tree with $k$ edges in $G$ can be augmented to a Steiner
  tree in $G'$ of weight $w$ with $(2m+1)k < w < (2m+1)k+2m$ by adding
  for cycle $C_T$ all except one of the edges.  Conversely, a Steiner
  tree in $G'$ of weight $w$ yields a Steiner tree in $G$ with
  $\lfloor w/(2m+1) \rfloor$ edges.  The reduction can be performed in
  polynomial time

  Finally, to get rid of the weights, we subdivide each edge of weight
  $w$ by $w-1$ subdivision vertices.  This shows that also the
  unweighted version is hard for these graphs.  Observe that this is a
  polynomial-time reduction since the weights above are polynomially
  bounded in the input size.
\end{proof}

\end{document}